\documentclass[12pt]{amsart}
\usepackage{amscd}
\usepackage{amsmath,amssymb,amsthm,mathrsfs,hhline}
\usepackage{yfonts}
\usepackage{hyperref}
\usepackage[dvips]{graphics}
\usepackage[centering,margin=1.2in]{geometry}
\usepackage[matrix,arrow,curve]{xy}

\newtheorem{thm}{Theorem}
\newtheorem{prop}{Proposition}
\newtheorem{lem}{Lemma}
\newtheorem{cor}{Corollary}
\newtheorem{remark}{Remark}

\theoremstyle{definition}

\theoremstyle{definition}
\newtheorem{defn}{Definition}

\newcommand{\bC}{\mathbf{C}}
\newcommand{\mT}{\mathcal{T}}

\newcommand{\Ad}{\text{Ad}}

\newcommand{\mf}{\mathfrak}

\newcommand{\g}{\mathfrak{g}}
\newcommand{\h}{\mathfrak{h}}

\newcommand{\be}{\mathfrak{b}}

\newcommand{\C}{\mathbb{C}}
\newcommand{\cO}{\mathcal{O}}

\newcommand{\cT}{\mathcal{T}}
\newcommand{\cM}{\mathcal{M}}
\newcommand{\cP}{\mathcal{P}}
\newcommand{\cB}{\mathcal{B}}
\newcommand{\cS}{\mathcal{S}}

\numberwithin{equation}{section}

\begin{document}

\title[Superintegrability of Symmetric Toda]{Superintegrability of Symmetric Toda}

\author{Nicolai Reshetikhin}
\address{N.R.: Department of Mathematics, University of California, Berkeley,
CA 94720, USA \& Physics Department, St. Petersburg University, Russia \& KdV Institute for Mathematics, University of Amsterdam,
Science Park 904, 1098 XH Amsterdam, The Netherlands.}
\email{reshetik@math.berkeley.edu}

\author{Gus Schrader}
\address{G.S.: Department of Mathematics, Columbia University , New York, NY 10027, USA}
\email{schrader@math.columbia.edu}

\begin{abstract}
In this paper we prove superintegrability of Hamiltonian systems generated
by functions on $K\backslash G/K$, restriced to a symplectic leaf of the Poisson variety $G/K$,
where $G$ is a simple Lie group with the standard
Poisson Lie structure, $K$ is the subgroup of fixed points with respect to the Cartan involution.
\end{abstract}

\maketitle

\section*{Introduction}

The emergence of Poisson brackets defined by the classical $r$-matrix \cite{Sk1,STS}
and the subsequent discovery of Poisson Lie groups \cite{Dr} were major milestones in the
modern understanding of many known integrable systems.

The basic relation between Poisson Lie groups and integrable systems is as follows.
For a quasitriangular Poisson Lie group (where Poisson brackets are defined by a classical $r$-matrix), the conjugation-invariant functions form a Poisson commutative subalgebra in the algebra of regular functions on the group.
The restrictions of these Poisson commuting functions to various symplectic leaves in the group thus provides a source of possible integrable systems.  However for a typical symplectic leaf, the rank of this commutative subalgebra is less than half of the dimension of the symplectic leaf, so the balance of dimensions required for Liouville integrability fails to be satisfied.

Nonetheless, when a quasitriangular Poisson Lie group is of the factorizable type, the Hamiltonian flow generated
by a central function can be described by the factorization formula \cite{STS}. For this reason, even in the case when the
number of independent central function is less that needed to ensure Liouville integrability,
one still expects that the corresponding systems will be superintegrable.

\begin{defn} A {\it superintegrable system} on a symplectic
manifold $(\cM_{2n}, \omega)$ consists of a Poisson subalgebra $C_J(\cM_{2n})\subset C(\cM_{2n})$ of rank $2n-k$
which has a Poisson center $C_I(\cM)$ of rank $k$.
\end{defn}

Superintegrability \cite{SInt} is known by several names in the literature. It is alternatively referred to as degenerate integrability\footnote{This terminology was introduced
by N. Nekhoroshev in \cite{N}. The first author was using this terminology in the previous papers.}, or as the non-commutative integrability\footnote{ This terminology was introduced by Mischenko and
Fomenko in the context of integrable systems related to Lie algebras and is frequently used in the literature,
see for example \cite{KM}\cite{FLGV} and other related papers.}. 

The Hamiltonian dynamics generated by the function $H\in C(\cM)$ is said to be superintegrable if $H\in C_I(M)$.
If $J_1,\dots, J_{2n-k}$ are independent functions from $C_J(\cM)$, we have
\[
\{H,J_i\}=0, \ \ i=1,\dots, 2n-k .
\]
In other words, functions from $C_J(\cM)$ are integrals of motion for $H$.
One can say that Hamiltonian vector fields generated by $J_i$ describe the symmetry of
the Hamiltonian flow generated by $H$.

Let  $I_1,\dots,I_k\in C_I(\cM)$ be $k$ independent functions (meaning that $dI_1\wedge\dots\wedge
dI_{k}$ does not vanish identically). 

\begin{thm}\label{deg}\cite{N}
\begin{enumerate}
\item The flow lines of $H$ are parallel to level surfaces of $J_i$.
\item Each connected component of a generic level surface
has a canonical affine structure generated by the flow lines of
$I_1,\dots, I_k$.
\item The flow lines of $H$ are linear in this affine
structure.

\end{enumerate}
\end{thm}

Geometrically, superintegrability corresponds to the existence of a pair of Poisson projections
\[
\cM\stackrel{\pi}{\rightarrow} \cP \stackrel{p}{\rightarrow} \cB
\]
where $\cP$ and $\cB$ are Poisson manifolds, $\pi$ and $p$ are
Poisson projections, and $\cB$ has trivial Poisson structure. 
Fibers of $p$ are finite unions of symplectic submanifolds (symplectic
leaves of $\cP$) and $dim(\cP)+dim(\cB)=dim(\cM)$.
In this geometric setting, we have $C_J(\cM)=\pi^* C(\cP)$ and $C_I(\cM)=(p\circ \pi)^*C(\cB)$.

When $k=n$ this theorem reduces to the Liouville integrability. In general, the difference between Liouville integrable systems on a symplectic $2n$ dimensional
manifold is that a superintegrable system may have Liouville tori of smaller dimension $k<n$. For more details on superintegrable systems, examples and references see \cite{R2,SInt, GS, LGV, KM}.

In the case of central functions on simple Lie groups with their standard Poisson Lie
structures, the superintegrability was proven in \cite{R1}. In this case symplectic leaves are isomorphic to quotients of double Bruhat cells by a torus.

The factorization formula for the Hamiltonian flow on a factorizable Poisson Lie group
was extended to Hamiltonian flows  on homogeneous spaces $G/K$ generated  by the Poisson commutative subalgebra of functions on $K\backslash G/K$   in \cite{gus}.
 Here $G$ is a Poisson Lie group,
and $K$ is the subgroup of fixed points of an involution on $G$ satisfying the classical reflection equation.
A particular case of
integrable systems on such homogeneous spaces are integrable systems with reflecting boundary conditions \cite{Sk2}.
The integrability of such systems when $G$ is the loop algebra of $SL_2$ in the general framework of integrable
systems on homogeneous spaces was established in \cite{Sch2}.

This paper is sequel to \cite{gus} and is also a sequel to \cite{R1}. Here we study the superintegrability of reflecting integrable systems arising from the Cartan involution on finite dimensional split real Lie groups with their standard
Poisson Lie structure. Specifically, we prove that for generic symplectic leaves of $G/K$, the Hamiltonian systems
generated by functions on $K\backslash G/K$ are superintegrable.

The plan of the paper is as follows. In sections \ref{P1} and \ref{P2} we recall the basic facts
about Poisson homogeneous spaces of a quasitriangular Poisson Lie groups $(G,r)$. In section \ref{NH} we prove integrability for Hamiltonian systems on symplectic leaves of $G/K$ with Hamiltonians from $K\backslash G/K$
and construct angle variables. In section \ref{CS} we relate the results of section \ref{NH} with
superintegrability of characteristic integrable systems on $G$.

\section{The standard Poisson structure simple Lie group and on their cosets.}\label{P1}
\subsection{Standard Poisson Lie structure on a simple Lie group} Let $G$ be a simple complex Lie group with Lie algebra $\g$.  Choosing a Borel subalgebra $\mf{b}_+\subset\g$ gives the corresponding root space decomposition $\g=\h+\bigoplus_{\alpha\in\Delta}\g^\alpha$. We have Lie subalgebras
$$
\mf{n}_+=	\bigoplus_{\alpha\in\Delta_+}\g^\alpha, \qquad \mf{n}_-=\bigoplus_{\alpha\in\Delta_-}\g^\alpha
$$
as well as $\mf{b}_{\pm}=\h+\mf{n}_{\pm}$.  We will use notations $N_{\pm}, B_\pm, H$ for the corresponding Lie subgroups in $G$, and $\langle,\rangle$ for a fixed non-zero multiple of the Killing form on $\g$. For each $\alpha\in \Delta^+$, let us fix root vectors $E_\alpha\in \g^\alpha, E_{-\alpha}\in \g^{-\alpha}$ normalized so that $\langle E_{\alpha},E_{-\alpha}\rangle=1$.
The canonical element $r\in \g\wedge\g$ defined by
\begin{align}
\label{r-matrix}
r=\sum_{\alpha\in\Delta_+}E_{\alpha}\wedge E_{-\alpha}
\end{align}
is known as the standard classical $r$-matrix for $\g$ and a choice of Borel $\mf{b}_+\subset \g$ \cite{Dr}.

\begin{remark} The $r$-matrix (\ref{r-matrix}) is the skewsymmetrization of the standard factorizable Lie bialgebra $r$-matrix
\[
r=\frac{1}{2} \sum h_i\otimes h^i+\sum_{\alpha\in\Delta_+}E_{\alpha}\otimes E_{-\alpha}
\]
This $r$-matrix satisfies classical Yang-Baxter equation, see for example \cite{STS}.
\end{remark}

Given a multivector $X\in \bigwedge^k\g$, denote by $X^L $ and $X^R$ the corresponding left and right invariant $k$-vector fields on $G$.  Then
$$
\eta= r^R-r^L
$$
defines a Poisson structure\footnote{As usual it means a Lie algebra structure
on the algebra of functions on $G$  with the Lie bracket $\{f,g\}=(\eta, df\wedge dg)$, where $(.,.)$ is the pairing between polyvectors and forms on $G$.}on $G$ called the {\em standard Poisson structure.}  With respect to the Poisson structure $\eta$, the group $G$ is a Poisson Lie group. If we use the trivialization of $TG$ by left translations as $TG\simeq \g\times G$ then
 \[
 \eta(x)=Ad_x(r)-r
 \]
where $Ad_x(a\wedge b)=Ad_x(a)\wedge Ad_x(b)$.

\subsection{Standard Poisson structure on cosets $G/K$}

The Lie algebra $\g$ admits a Lie algebra involution $\sigma$ called the {\em Cartan involution}, defined by
\begin{align*}
&\sigma(E_{\alpha})=-E_{-\alpha}, \ \alpha\in \Phi
\\
&\sigma(H)=-H, \ H\in\mf{h}
\end{align*}
 The involution $\sigma$ lifts to a corresponding Lie group involution (involutive automorphism), which we denote by the same letter.
 \begin{lem}
 \label{anti}
 The involution $\sigma$ is anti-Poisson with respect to $\eta$.
 \end{lem}
 \begin{proof}
 This follows from the action of $\sigma$ on the classical $r$-matrix: $\sigma^{\otimes 2}(r)=-r$.
 \end{proof}
The Cartan involution gives rise to a decomposition $\g=\mathfrak{k}\oplus\mathfrak{p}$ into its $\pm 1$ eigenspaces known as the {\em Cartan decomposition}. The fixed point set $\mathfrak{k}$ is a Lie subalgebra in $\g$, the anti-fixed point set $\mathfrak{p}$ is a $\mathfrak{k}$-module. We denote by $K$ the Lie subgroup of $G$ corresponding to the Lie subalgebra $\mathfrak{k}$.
\begin{prop}
\label{cre-prop} The pair $(r,\sigma)$ is a solution of the classical reflection equation
\begin{align}
(\sigma\otimes\sigma)(r)+r=(\sigma\otimes1)r-(1\otimes\sigma)r
\end{align}
\end{prop}
In fact, as observed in Lemma \ref{anti}, $\sigma$ satisfies $(\sigma\otimes\sigma)(r)=-r$, so both sides of the classical reflection equation vanish.  It follows from Proposition \ref{cre-prop} that $\mathfrak{k}$ is a coideal Lie subalgebra in $\g$.    Hence the Poisson structure $\eta$ on $G$ descends to a well-defined Poisson structure $\eta_{G/K}$ on the coset space $G/K$.  Note, however, that $K$ is not a Poisson-Lie subgroup in $G$.

\section{Poisson structure on the symmetric space $G/K$. }\label{P2}It is clear that the standard Poisson Lie structure on a simple complex Lie group $G$ descends to its
split real form. {\it From now on $G$ means the split real form of $G_{\mathbb C}$ and $K$ means the corresponding compact real form
of $K_{\mathbb{C}}$ from the previous section. We will denote by $AN\subset B_+$ the Lie subgroup where $N\subset B_+$ is the unipotent subgroup in the
Borel subgroup and $A$ is the positive subgroup of the split real form of the Cartan subgroup $H\subset G$.}

In this section we describe some properties of the Poisson structure $\eta_{G/K}$ on $G/K$.  Recall \cite{helgason} that $G$ admits an Iwasawa decomposition: the multiplication map $AN\times K\rightarrow G$ is a real diffeomorphism.   We may therefore identify the coset space $G/K$ with $AN$.  Observe that $AN\subset B_+\subset G$ are Poisson submanifolds.  Let us recall the following results of \cite{gus} regarding $\eta_{G/K}$:
\begin{prop}
The subalgebra $\bC[K\backslash G/K]$ of $K$-invariant functions on $G/K$ is Poisson commutative.
\end{prop}
We shall refer $\bC[K\backslash G/K]$ as the Poisson subalgebra of {\em reflection Hamiltonians.}
Elements of it are functions on $G$ of the form $f(g\sigma(g)^{-1})$ where $f$ is such that $f(hgh^{-1})=f(g)$.
It is clear that these functions are left and right $K$-invariant.

\begin{prop}
\label{bpoissonstr}
The Poisson structure $\eta_{G/K}$ on $G/K$ coincides with the Poisson structure on $B_+$ coming from its inclusion into $G$.
\end{prop}
\begin{remark} Both propositions remain true in the complex case.
\end{remark}

Recall that if $u$ is an element of the Weyl group $W$, the double Bruhat cells $G^{1,u}=B_+\cap B_-uB_-$  are Poisson subvarieties of $B_+$. Each double Bruhat cell $G^{1,u}$ is
fibered over the torus $\ker( u-\mathrm{id})\subset H$, with fibers being symplectic leaves of dimension $l(u)+\mathrm{corank}(u-\mathrm{id})$, see \cite{HL}. The fibration is given by certain generalized minors; for details see \cite{KZ}, \cite{hkkr},\cite{R1}. This gives the description of symplectic leaves of $B_+$. They restrict to symplectic leaves of $AN$.

Combining Proposition \ref{bpoissonstr} with the description of symplectic leaves of $B_+$ we obtain
\begin{prop}
The symplectic leaves of $G/K$ after identification of this space with $AN$ coincide 
with the intersection symplectic leaves of homogeneous Poisson varieties $B_+\cap B_-uB_-$ with $AN$.
\end{prop}
Here it is essential that we have a global isomorphism $G/K\simeq AN$, i.e. this proposition holds for
split real form.

As a consequence,  the restriction of the reflection Hamiltonians to a symplectic leaf $S$ form Poisson commutative subalgebra $I_S$ in $\bC[S]$. In the next section we will show that this subalgebra defines a 
superintegrable system.

\bigskip
%

\section{Poisson noncommutative Hamiltonians and Degenerate Integrability }\label{NH}
\subsection{Non-comutative Hamiltonians}

In this section we explain how to construct a Poisson subalgebra $\mathcal{A}$ of $\C[G/K]$ that Poisson commutes with the reflection Hamiltonians $\C[K\backslash G/K]$.

Define the mapping
\begin{align}
\label{reflection-monodromy}
\mathcal{T}\colon G \to G, \qquad g\mapsto g\sigma(g^{-1}).
\end{align}
It descends to define the {\em reflection monodromy mapping}
\[
\widehat{\cT}: G/K\to G, \ \ [g]\mapsto g\sigma(g^{-1})
\]

\begin{remark}
The mapping $\widehat{\cT}$ induces an isomorphism between $K\backslash G/K$
and the semisimple part of the coset space $(G/Ad_G)_{>0}$ of orbits passing through
elements of $G$ with positive principal generalized minors. In other words
it identifies the space of functions on $K\backslash G/K$ with $G$-invariant functions
on $G_{>0}$ (elements of $G$ with positive principal generalized minors).
\end{remark}

For convenience, let us introduce elements $Y_{\alpha}\in\mf{k}, \ \alpha\in \Phi^+$ defined by
$$
Y_\alpha := E_{-\alpha}-E_{\alpha}
$$
Note that we may write
\begin{align}
\label{iwasawar}
r&=\sum_\alpha E_\alpha\wedge Y_\alpha
\end{align}

Let $\cT_*$ be the mappings between vector fields on $G$.
\begin{prop} The mapping $\mathcal{T}$ satisfies
\begin{align}
\label{derivatives}
\nonumber\mathcal{T}_* Y_\alpha^L&=0\\
\mathcal{T}_* Y_\alpha^R&=Y_\alpha^R-Y_\alpha^L\\
\nonumber\mathcal{T}_* E_{\alpha}^R&=E_{\alpha}^R+E_{-\alpha}^L
\end{align}
Here and below $X^{L,R}$ are left and right vector fields on $G$.
\end{prop}

We will also use the involutive diffeomorphism $\tau$ of $G$ defined by

$$
\tau(g)=\sigma(g^{-1})
$$
By $\tau_*$ we again denote the mapping between vector fields on $G$.
\begin{lem}
The involution $\tau$ is a group anti-automorphism, an automorphism of Poisson varieites, and $\tau(B_\pm)=B_\mp$.  Moreover, for all $\alpha\in\Phi$, we have
\begin{align}
\label{tau-deriv}
\tau_* E_{\alpha}^L=E_{-\alpha}^R
\end{align}
and
\begin{align}
\label{tau-mt}
\tau\circ\mT=\mT
\end{align}
\end{lem}
The proof is clear; for example one has

\[
\tau(\cT(g))=\tau^2(g)\tau(g)=g\tau(g)=\cT(g)
\]

Let $\tau^* \colon \C[G]\rightarrow \C[G] $ denote the pull-back of functions under the mapping $\tau$.
\begin{prop}\label{T-br}
The reflection monodromy mapping $\mathcal{T}$ satisfies
\begin{align}
\label{rm-pb}
\{\mT^*{f_1}, \mT^*f_2 \}=\frac{1}{2}\mT^*(\{ f_1+ \tau^* f_1, f_2+ \tau^* f_2\})
\end{align}
for all $f_1,f_2\in \C[G/K]$.
\end{prop}
\begin{proof}
Let us compute the Poisson brackets of the two elements $\mathcal{T}^*f_1,\mathcal{T}^*f_2$ of $\mathcal{A}$, where $f_i\in\C(G/\Ad_{B_+})$.  We have
\begin{align*}
\{\mathcal{T}^*f_1,\mathcal{T}^*f_2\}(g)&=\langle (\mathcal{T}^*)^{\otimes 2}(df_1\otimes df_2)_{g}, {r}^R-{r}^L\rangle\\
&=\langle (df_1\otimes df_2)_{g\sigma(g^{-1})},\mathcal{T}_*^{\otimes 2}\left({r}^R-{r}^L\right) \rangle\\
&=\langle (df_1\otimes df_2)_{g\sigma(g^{-1})},(E_\alpha^R+E_{-\alpha}^L)\wedge(Y_\alpha^R-Y_\alpha^L)\rangle\\
&=\langle (df_1\otimes df_2)_{g\sigma(g^{-1})},E_\alpha^R\wedge E_{-\alpha}^R-E_\alpha^L\wedge E_{-\alpha}^L\rangle\\
& \ \ \ +\langle (df_1\otimes df_2)_{g\sigma(g^{-1})},E_\alpha^R\wedge E_\alpha^L+E_{-\alpha}^L\wedge E_{-\alpha}^R\rangle \\
&=\{f_1,f_2\}(\mT(g))+\langle (df_1\otimes df_2)_{g\sigma(g^{-1})},E_\alpha^R\wedge E_\alpha^L+E_{-\alpha}^L\wedge E_{-\alpha}^R\rangle
\end{align*}
But now in view of properties (\ref{tau-mt}) and (\ref{tau-deriv}) of the Poisson automorphism $\tau$, we may re-express
\begin{align*}
&\sum_{\alpha\in \Delta_+}\langle (df_1\otimes df_2)_{g\sigma(g^{-1})},E_\alpha^R\wedge E_\alpha^L+E_{-\alpha}^L\wedge E_{-\alpha}^R\rangle\\
&\qquad=\frac{1}{2}\{\tau^*f_1,f_2\}(\mT(g)) +\frac{1}{2}\{f_1,\tau^*f_2\}(\mT(g))
\end{align*}
and thus the result follows.
\end{proof}

Now let $\mathcal{A}=\mathcal{T}^*\C(G/\Ad_{B_+})$ be the pullback of the Poisson subalgebra of $\Ad_{B_+}$-invariant functions on $G$ under the mapping $\mathcal{T}$. Note that since $\tau\circ\mathcal{T}=\mathcal{T}$, we have that
$$
\mathcal{A}=\mathcal{T}^*\C(G/\Ad_{B_+})=\mathcal{T}^*\C(G/\Ad_{B_-}).
$$

\begin{cor}
The Poisson subalgebra $\mathcal{T}^*\C(G/\Ad_G)$ of reflection Hamiltonians Poisson commutes with $\mathcal{A}=\mathcal{T}^*\C(G/\Ad_{B_+})$
\end{cor}
\begin{proof}
Recall that the Poisson subalgebra $\C(G/\Ad_G)$ Poisson commutes with both subalgebras $\C(G/\Ad_{B_+}), \C(G/\Ad_{B_-})$; this follows immediately from the fact that the quasitriangular $r$-matrix $r$ satisfies $r\in \be_+\otimes\be_-\subset\g\otimes\g$.  Hence the assertion of the corollary follows by taking $f_1\in \C(G/\Ad_G)$ and $f_2\in \C(G/\Ad_{B_+})$ in formula (\ref{rm-pb}), by definition of
$\tau$:
\[
\tau^*(f_1)(hgh^{-1})=f_1(\sigma(hg^{-1}h^{-1}))=f_1(\sigma(h)\sigma(g)^{-1}\sigma(h)^{-1})=\tau^*(f_1)(g)
\]
Thus, $\tau^*(f_1)\in C^G(G)$. Simlarly $\tau^*(f_2)\in C^{B_-}(G)$. Because Poisson subalgebra
$G^G(G)$ Poisson commutes with subalgebras $C^{B_\pm}(G)$ \cite{R1}, Proposition \ref{T-br} implies
\[
\{\cT^*(f_1),\cT^*(f_2)\}=\frac{1}{2}\cT^*\{f_1+\tau^*(f_1), f_2+\tau^*(f_2)\}=0
\]

\end{proof}

Thus, we have the following embedding of Poisson algebras
\begin{equation}
\mT^*(C^G(G))=C(K\backslash G/K)\subset \mathcal{T}^*\C(G/\Ad_{B_+})=\mathcal{T}^*\C(G/\Ad_{B_-})\subset C(G/K),
\end{equation}and therefore the subalgebra $\mathcal{A}$ is a natural candidate for the Poisson subalgebra
which guarantees superintegrability of reflection Hamiltonians, i.e. elements of $\mT^*(C^G(G))$.
These embedding correspond to Poisson projections

\begin{equation}\label{Pam}
\mathcal{T}(B)\rightarrow \mathcal{P}\rightarrow G/Ad_G
\end{equation}
Here $\mathcal{P}$ is the set of all $Ad_{B_+}$-orbits through $\mT\left(B_+\right)\subset G$.

The restriction of the Poisson projections \eqref{Pam}
to a symplectic leaf corresponding to $u\in W$  gives a sequence of Poisson projections:
\begin{equation}\label{Pm}
\cS^u\to \cP^u\to \cB^u.
\end{equation}
Here $S^u\subset AN=G/K$ is a symplectic leaf of $G/K$ corresponding to the double
Bruhat cell $G^{1,u}_+=(B_+\cap B_-uB_-)\cap AN\subset AN\subset  B_+$, while $\cP^u$ is the set of all $Ad_{B_+}$-orbits in $G$
through $\cS^u\subset \mT\left(B_+\right)\subset G$. Finally, the subset $\cB^u\subset G/Ad_G$ is the set of all $Ad_G$-orbits through $\cS^u\subset \mT\left(B_+\right)\subset G$.

\subsection{Degenerate integrability}
Now we will prove that \eqref{Pm} is a superintegrable system. For this we need to prove the balance of dimensions in \eqref{Pm}.
Here and below we assume that $S^u$ is generic, in the sense that $dim(\cB^u)=r$. 


\begin{thm}
For generic $u$ we have
\[
dim(\cP)=dim(\mathcal{T}(AN))-r
\]
\end{thm}
\begin{proof}
Let $b\in AN$ be a generic element, in the sense that $b\tau(b)\in G$ is regular (conjugate to a 
generic point in $H\subset G$). For given $b$ let us describe all $b'$ such that 
\begin{equation}\label{rep}
\beta b\tau(b)\beta^{-1}=b'\tau(b'),
\end{equation}
for some $\beta\in B_+$. 
If $\cO_{b\tau(b)}$ is the $Ad_{B_+}$-orbit through $b\tau(b)$ such $b'$ describe intersection points 
$\cO_{b\tau(b)}\cap \mathcal{T}(AN)$. We will show that $dim(\cO_{b\tau(b)}\cap \mathcal{T}(AN))=r$.
This implies the desired equality.

Applying $\tau$ to (\ref{rep}) we obtain
\[
\tau(\beta)^{-1}b\tau(b)\tau(\beta)=b'\tau(b')
\]
or
\begin{equation}\label{comm}
\tau(\beta)\beta b\tau(b)=b\tau(b)\tau(\beta)\beta
\end{equation}
Because $b$ is generic there exists $U\in K$ such that
\[
b\tau(b)=U\Lambda U^{-1}
\]
Combining this with (\ref{comm}), we conclude that
\[
\tau(\beta)\beta=UDU^{-1}
\]
where $D\in H$. In fact, since $\beta\in AN$ we have $D\in A\subset H$.
Because $b\mapsto b\tau(b)$ is a diffeomorphism $AN\to G/K$, the choice of $D$ determines
$\beta$ and therefore $b'$ uniquely for a given $b$. This proves
\[
\cO_{b\tau(b)}\cap \mathcal{T}(AN)\simeq A
\]
In particular, its dimension is $r$. This concludes the proof.
\end{proof}

In $AN$ and $(AN)^u=AN\cap B_-uB_-$, for generic $u$, regular elements form a Zariski
open subset. Choose $b\in (AN)^u\subset AN\subset B_+$ to be semisimple.
By the same arguments as above we obtain the proof of the following statement.

\begin{thm}
For generic $u$ we have
\[
dim(\widetilde{\cP}^u)=dim\left(\mathcal{T}((AN)^u)\right)-r
\]
where $\widetilde{\cP}^u$ is the set of $Ad_{B_+}$-orbits through $(AN)^u\subset G$.
\end{thm}

Symplectic leaves of $(AN)^u$ are level sets of generalized minors \cite{KZ}. 
For generic $u$ regular elements form Zariski open subset in these level sets
as well. Thus we have

\begin{cor}
For generic $u$ and a symplectic leaf $S^u$ in $(AN)^u$ 
\[
dim(\cP^u)=dim(\mathcal{T}(\cS^u))-r
\]
Here $\cP^u$ is the set of $Ad_{B_+}$-orbits through $S^u\subset G$.
\end{cor}

%
%
%

This proves the superintegrability of Hamiltonian systems on generic symplectic leaves
of $G/K$ with Hamiltonians from $C(K\backslash G/K)\simeq C(G/Ad_G)$.  When $u$ is a Coxeter element, the corresponding integrable system is a non-degenerate, i.e. Liouville integrable system, isomorphic to the relativistic
Toda system corresponding to this Lie algebra \cite{hkkr}. If the Weyl group element $u$ is non-generic, there are nontrivial stabilizers, but in a similar way one can show that the restriction of \ref{Pam} to such symplectic leaf still
gives a superintegrable system. In this case one should repeat the arguments above for the 
diagrammatically embedded semisimple subgroup in $G$ where $u$ is generic.

\subsection{Action-angle variables}
For general background on action angle variables, see \cite{N}\cite{FLGV} and \cite{KM}. Here we will explicitly describe
action angle variables for our systems.

Suppose that $V_\lambda$ is an irreducible real representation of the split real group $G$ that contains a spherical vector, that is a vector $u\in V_\lambda$ satisfying $ku=u$ for all $k\in K\subset G$. Such representations and vectors can be constructed as follows.  The real representation $\mathrm{Sym}^2\left( V_\lambda\right)$ can be regarded as the space of real quadratic forms on $V_\lambda^*$.  Then one can take any positive definite form and average it over the compact group $K$, to obtain a nonzero spherical vector $u\in \mathrm{Sym}^2\left( V_\lambda\right)$.  For example, in the case $G=SL_n(\mathbb{R})$ and its first fundamental representation $V_{\omega_1}\simeq \mathbb{R}^n$, we have $\mathrm{Sym}^2 V_{\omega_1}\simeq V_{2\omega_1}$, and we obtain a spherical vector $u\in V_{2\omega_1}$ fixed by $K=SO_n(\mathbb{R})$.

 Now let $g(t)$ be the Hamiltonian flow generated by $H\in G(K\backslash G/K)$
passing through $g_0=b_0\tau(b_0)\in S^u$. It was shown in \cite{gus} that the  Hamiltonian flow line generated by the function $H$ passing through $g_0\in G$ at $t=0$
can be described as
\begin{equation}\label{evol}
g(t)=k_+(t)^{-1}g_0k_-(t),
\end{equation}
where $k_\pm(t)\in K$ are defined by the factorization
\[
\exp(t\nabla^\pm H(g))=b(t)k_\pm(t)^{-1},
\]
in which $b(t)\in AN$, $\nabla^\pm H$ are left and right gradients of $H$, and $\langle\nabla^\pm H, X\rangle$ are left and right derivative of $H$
with respect to $X\in \g$, so for example $\langle\nabla^- H(g), X\rangle=\frac{d}{dt}H(ge^{tX})|_{t=0}$. Here $\langle\cdot,\cdot\rangle$ is the Killing form on $\g$.


Consider the spectral decomposition
\[
g=\sum_{\alpha} h^\lambda_\alpha Q^\lambda_\alpha
\]
for a generic element $g\in G$ acting in an a finite dimensional irreducible representation $V_\lambda$ with the highest weight $\lambda$.
Note that $\nabla^-H(g)$ is diagonal on $Q_\alpha$: $\nabla^-H(g)Q_\alpha=f(h_\alpha)Q_\alpha$. This is most obvious for matrix groups
when $\tau$ is the transposition, the Killing form is a trace and $H(g)=tr(g)$.

From the formula for the evolution of $g(t)$ we conclude that $h^\lambda_\alpha$ are preserved by the evolution
and that the idempotents $Q^\lambda_\alpha$ evolve as
\[
Q^\lambda_\alpha(t)=k_-(t)^{-1}Q^\lambda_\alpha k_-(t),
\]
and therefore
\begin{equation}\label{Q-evol}
Q^\lambda_\alpha(t)=b(t)^{-1}\exp(t\nabla^\pm H(g_0))Q^\lambda_\alpha k_-(t).\end{equation}

Let $(u,v)$ be the Shapovalov form on $V_\lambda$. It is uniquely defined by the
normalization $(v_\lambda, v_\lambda)=1$ where we fixed a highest weight vector $v_\lambda$ and
the property $(\tau(X)v,u)=(v, Xu)$. Consider variables
\[
r_\alpha^\lambda=(v_\lambda, Q_\alpha u),
\]
where $v_\lambda$ is a lowest weight vector and $u\in V_\lambda$ is a
spherical vector, that is $ku=u$ for all $k\in K\subset G$.
From \eqref{Q-evol} we obtain
\[
r_\alpha^\lambda (t)=h_\lambda(t)^{-1}e^{t f(h_\alpha)}r_\alpha^\lambda,
\]
where the function $h_\lambda(t)$ describes the action of $b(t)$ on the lowest weight vector: $\tau(b(t))^{-1}v_\lambda=h_\lambda(t)^{-1}v_\lambda$.
From here we conclude that
\[
\frac{r_\alpha^\lambda(t)}{r_\beta^\lambda(t)}=\exp\left(t(f(h_\alpha)-f(h_\beta))\right) \frac{r_\alpha^\lambda}{r_\beta^\lambda}.
\]

Thus we proved that the ratios $\frac{r_\alpha^\lambda}{r_\beta^\lambda}$ evolve logarithmically linearly and therefore
form angle variables.

\section{The relation to integrability of characteristic systems on $G$.}\label{CS}

Let us now consider Poisson variety $G/\tau$ obtained as the quotient of $G$ by the Poisson automorphism $\tau$.  The ring of functions on $G/\tau$ is identified with the Poisson subalgebra $\C[G/\tau]=\C[G]^\tau \subset \C[G]$ of $\tau$-invariant functions on $G$. Note that Proposition \ref{rm-pb} has the following simple corollary:
\begin{cor}
If $f_1,f_2\in \C[G/\tau]$, then
$$
\{\mT^*f_1,\mT^*f_2\}=2\mT^* \{f_1,f_2\}
$$
Hence, rescaling the Poisson strucure on $G/K$ by a factor of $\frac{1}{2}$, the map $(G/K,\eta_{G/K}/2)\hookrightarrow G \rightarrow G/\tau$ obtained by composing $\mT$ with the quotient projection is Poisson.
\end{cor}

Let $D(G)\simeq G\times G$ be the Poisson-Lie double of $G$ \cite{Dr,STS}.  If $\{T_i\}$ is an orthonormal basis for $\h$ with respect to the Killing form $\langle,\rangle$, the Poisson tensor of $D(G)$ is defined in terms of the canonical element
$$
\Lambda=\sum_{i}(T_i,-T_i)\wedge(T_i,T_i)+\sum_{\alpha\in\Phi^+}\left((E_\alpha,0)\wedge(E_{-\alpha},E_{-\alpha})+(0,E_{-\alpha})\wedge(E_{\alpha},E_\alpha).\right)
$$
The Poisson bivector field is
$$
\eta_{D(G)}=\Lambda^R-\Lambda^L,
$$
where we write $(x,y)$ for elements of the Lie algebra $Lie(D(G))=\g\oplus \g$.
We consider the involution $s$ on $D(G)$ defined by $s(g_1,g_2)=(\tau(g_2),\tau(g_1))$.
\begin{lem}
The map $s$ is a Poisson automorphism of $D(G)$, which descends to a well-defined Poisson automorphism of the quotients  $D(G)/\Ad_{B_+\times B_-}$ and $D(G)/\Ad_{D(G)}$.
\end{lem}
Indeed, $s$ preserves $B_+\times B_-\subset G\times G$ and the diagonally embedded $G$. Now, we have a well-defined diagonal embedding
$$
\Delta: G/\tau \rightarrow D(G)/s, \quad [g]\mapsto [(g,g)]
$$
Indeed, $f\in \C[D(G)/s]$ means that $f(\tau(g_2),\tau(g_1))=f(g_1,g_2)$ so that in particular $f(\tau(g),\tau(g))=f(g,g)$, which implies the statement.
Since  $\tau$ and $s$ are Poisson automorphisms and the diagonal embedding of $G$ into $D(G)$ is Poisson, so is that of $G/\tau$ into $D(G)/s$.
\begin{cor}
We have the following diagram in which all maps are Poisson:
\end{cor}
$$
\xymatrix@C+=2.9cm{
(G/K,\frac{1}{2}\eta_{G/K})\ar[r]\ar[d]^{\pi}& G/\tau \ar[r]^{\Delta}\ar[d]^{\Delta}& D(G)/s\ar[d]\\
\cP \ar[d]^{\theta} \ar[r]^{\mu} & \cP' \ar[r]\ar[d]^{\psi}& \left(D(G)/\Ad_{B_+\times B_-}\right)/s\ar[d]\\
K\backslash G/K \ar[r] & (G/Ad_G)/\tau \ar[r]^{\Delta}& \left(D(G)/\Ad_{D(G)}\right)/s
}
$$
Here $\cP$ is the space of $Ad_{B_+}$-orbits in $G$ through $\mT\left(B_+\right)$.
Equivalently, this is the space of $Ad_{B_-}$ such orbits, and $\cP'$ is the set of equivalence classes
$$
\cP'=\left\{[g\tau(g),g\tau(g) ]  | g\in G\right\} \subset  \left(D(G)/\Ad_{B_+\times B_-}\right)/s.
$$
The map $\pi$ is a natural projection, and
$\theta(Ad_{B_+}(b\tau(b)))=Ad_G(b\tau(b))$, while the left bottom and top horizontal arrows  are the natural embeddings. Note that $\tau$ acts trivially on semisimple elements of $G/Ad_G$ and recall that
the mapping $K\backslash G/K\simeq G/Ad_G\to G/Ad_G$ acts as $KgK\mapsto Ad_G\left(g\tau(g)\right)$.

In particular, the diagram implies that the superintegrability of Hamiltonian systems on $G/K$ can in fact be deduced from
the superintegrability of characteristic systems on $G$ established in \cite{R1}.

\section{Conclusion}

In this note we proved the superintegrability of reflecting integrable systems arising from the Cartan involution on finite dimensional split real Lie groups with their standard
Poisson Lie structure. We anticipate that similar results hold for other Poisson Lie groups.

Namely, we expect the following:

\begin{itemize}

\item {\it For loop groups with the standard Poisson Lie structure, Hamiltonian systems
generated by Poisson commutative subalgebra of $G$-invariant functions (or their
twisted version) are superintegrable on generic symplectic leaves}.
For loops in $SL_n$ the Liouville integrability is proved for symplectic
leaves corresponding to cyclically reduced elements of the corresponding affine Weyl group
\cite{FM}. These integrable systems are isomorphic to the ones constructed
in \cite{GK} on cluster varieties corresponding to dimers \cite{FM}. On symplectic
leaves corresponding to other elements of the affine Weyl group \cite{HW} one should expect superintegrability.

\item {\it Classical spin chains generated by Hamiltonians with "reflection" boundary conditions
are superintegrable on generic symplectic leaves of the corresponding affine homogeneous
spaces.} For the $LSL_2$ case the Lioville integrability on symplectic leaves corresponding
to  XXZ spin chains with reflection boundary conditions was analyzed in \cite{Sch2}.

\item {\it For Poisson Lie groups with non-standard, Belavin-Drinfeld Poisson Lie structure
and on corresponding symmetric homogeneous spaces, Hamiltonian systems generated by
natural Poisson commutative subalgebras are superintegrable for generic
symplectic leaves and Lioville integrable only for special ones.}

\end{itemize}

The authors are grateful to M. Semenov-Tian-Shansky for valuable remarks and questions about the first draft of the
paper.


\begin{thebibliography}{0}

\bibitem[Dr] {Dr}
Drinfeld, V.:  {\it Quantum groups}. In Proc. Intern. Congress of Math. (Berkeley
1986), pages 798–820. AMS, 1987.

\bibitem[H] {helgason}
Helgason, S.:
{\em Differential Geometry, Lie Groups, and Symmetric Spaces. } Graduate Studies in Mathematics, AMS 2001.

\bibitem[HL]{HL} T. Hodges and T. Levasseur. Primitive ideals of $C_q[SL(3)]$. Commun. Math. Phys., 156:581, 605, 1993.

\bibitem[HKKR]{hkkr}
Hoffmann, T.,  Kellendonk, J., Kutz, N. and Reshetikhin,N.:
{\em Factorization dynamics and Coxeter-Toda lattices},
Commun. Math. Phys. {\bf 212} (2000), 297-321, {\tt arXiv:solv-int/9906013.}

\bibitem[FLGV]{FLGV}   Fernandes, R.L., Laurent-Gengoux, C., Vanhaecke, P., Global Action-Angle variables for Non-Commutative Integrable Systems, arXiv:1503.00084.
    
\bibitem[FM]{FM}
Fock, V.V., Marshakov, A.: {\it Loop groups, Clusters, Dimers and Integrable systems}, {\tt arXiv:1401.1606.}

\bibitem[GS]{GS}  M.I. Gekhtman, M.Z. Shapiro. Non-commutative and commutative integrability of generic Toda flow in
simple Lie algberas. Comm. Pure Appl. Math. 52: 53–84 (1999). 

\bibitem[GK]{GK}
Goncharov, A.B., Kenyon, R.: {\it Dimers and cluster integrable systems}, {\tt  arXiv:1107.5588.}

\bibitem[KM]{KM} Kiesenhofer, A., Miranda, E., Non-commutative integrable systems on b-symplectic manifolds, Regul. Chaotic Dyn. 21 (2016), no. 6, 643–659.

\bibitem[KZ]{KZ} Kogan, Mikhail, and Andrei Zelevinsky. "On symplectic leaves and integrable systems in standard complex semisimple Poisson-Lie groups." International Mathematics Research Notices 2002.32 (2002): 1685-1702.

\bibitem[N]{N}  Nekhoroshev, N.N., Action-angle variables and their generalizations. Trans. Moscow Math. Soc. 26:180-
197 (1972).

\bibitem[R1]{R1}
N.~Reshetikhin:
{\em Integrability of characteristic Hamiltonian systems on
simple Lie groups with standard Poisson Lie structure}, . Comm. Math. Phys. {\bf 242} (2003), no. 1-2, 129, {\tt arXiv:math/0103147.}

\bibitem[R2]{R2} N. Reshetikhin: {\it  Degenerately Integrable Systems, preprint}, {\tt arXiv:1509.00730.}

\bibitem[Sch1]{gus}
G.~Schrader:
{\em Integrable systems from the classical reflection equation}  {\tt arXiv:1405.5506.}

\bibitem[Sch2]{Sch2} G. Schrader
{\it Algebraic integrability of the classical XXZ spin chain with reflecting boundary conditions},
{\tt ArXiv:1408.5200.}

\bibitem[SInt]{SInt} Superintegrability in Classical and Quantum Systems, Edited by: P. Tempesta, P. Winternitz, J. Harnad, W. Miller, Jr., G. Pogosyan, M. Rodriguez, CRM Proceedings and Lecture Notes, Volume: 37, 2004.

\bibitem[Sk1]{Sk1}
Sklyanin E.: {\it Quantum Inverse Scattering Method}, Zap. Nauch. Semin, LOMI, {\bf 95} (1980), 55-128.

\bibitem[Sk2]{Sk2}
Sklyanin E. :{\it  Boundary conditions for integrable quantum systems}, J. Phys.
A: Math. Gen. {\bf 21} (1988) 2375–2389.

\bibitem[STS]{STS} M. Semenov-Tian-Shansky :
\textit{Dressing transformations and Poisson group actions},
Publ.\ Res.\ Inst.\ Math.\ Sci.\ \textbf{21} (1985), 1237--1260.

\bibitem{HW}[W] Williams, H., Double Bruhat Cells in Kac–Moody Groups
and Integrable Systems,
Lett Math Phys (2013) 103:389–419.


\end{thebibliography}
\end{document}